\documentclass[authoryear,review]{article}
\usepackage{appendix}
\usepackage{setspace}
\usepackage[font=small,  margin=2cm]{caption}
\usepackage[tbtags]{amsmath}
\usepackage{amsthm,amssymb,amsfonts}
\usepackage[numbers]{natbib}
\usepackage{multirow}
\usepackage{subcaption}
\usepackage{lscape}
\usepackage{makecell}
\usepackage{exscale}
\usepackage{booktabs}
\usepackage{array}
\usepackage{fullpage}
\usepackage{url}
\usepackage{algorithm}
\usepackage{algpseudocode}
\usepackage{bm}
\usepackage{smile}
\usepackage{mathtools}
\usepackage{wrapfig}
\usepackage{lipsum}
\usepackage{mathrsfs}
\usepackage{relsize}
\usepackage{dsfont}
\usepackage{multirow}
\usepackage{apalike}
\usepackage{chngcntr}
\usepackage[usenames,dvipsnames,svgnames,table]{xcolor}

\newcommand{\bmax}{\mathop{\mathrm{max}}}
\newcommand{\bmin}{\mathop{\mathrm{min}}}
\newcommand{\bargmin}{\mathop{\mathrm{arg\ min}}}
\numberwithin{equation}{section}
\numberwithin{theorem}{section}
\numberwithin{corollary}{section}
\counterwithout{asmp}{section}
\numberwithin{definition}{section}

\begin{document}

\title{\LARGE Robust Covariance Estimation for High-dimensional Compositional Data  with Application to Microbial Communities Analysis}

	\author{Yong He\thanks{Institute for Financial Studies, Shandong University, Jinan, China; Email:{\tt heyong@sdu.edu.cn}.},~~Pengfei Liu\thanks{School of Mathematics and Statistics, Jiangsu Normal University, Xuzhou,  China; Email:{\tt liupengfei@jsnu.edu.cn}.},~~Xinsheng Zhang\thanks{ School of Management, Fudan University, Shanghai, China; Email:{\tt xszhang@fudan.edu.cn}.},~~Wang Zhou\thanks{Department of Statistics and Applied Probability, National University of Singapore, Singapore; Email:{\tt stazw@nus.edu.sg }.}}	
	\date{}	
	\maketitle
Microbial communities analysis is drawing growing attention due to the rapid development of high-throughput sequencing techniques nowadays. The observed data has the following typical characteristics: it is high-dimensional,     compositional (lying in a simplex) and even would be leptokurtic and highly skewed due to the existence of overly abundant taxa, which makes the conventional correlation analysis infeasible to study the co-occurrence and co-exclusion relationship between microbial taxa. In this article, we address the challenges of covariance estimation for this kind of data. Assuming the basis covariance matrix lying in a well-recognized class of sparse covariance matrices, we adopt a proxy matrix known as centered log-ratio covariance matrix in the literature, which is approximately indistinguishable from the real basis covariance matrix as the dimensionality tends to infinity. We construct a Median-of-Means (MOM) estimator for the centered log-ratio covariance matrix and propose a thresholding procedure that is adaptive to the variability of individual entries. By imposing a much weaker finite fourth moment condition compared with the sub-Gaussianity condition in the literature, we derive the optimal rate of convergence under the spectral norm. In addition, we also provide theoretical guarantee on  support recovery.  The adaptive thresholding procedure of the MOM estimator is easy to implement and gains robustness when outliers or heavy-tailedness exist. Thorough simulation studies are conducted to show the advantages of the proposed procedure over some state-of-the-arts methods. At last, we apply the proposed method to analyze a microbiome dataset in human gut. The R script for implementing the method is available at \url{https://github.com/heyongstat/RCEC}.

\vspace{2em}

	
\section{Introduction}

Covariance matrix estimation plays an important role in many areas of statistical analysis such as Principle Component Analysis (PCA), Linear Discriminant Analysis (LDA) and Gaussian Graphical Models (GGM). Nowadays, rapid development in computer technology floods us with high-dimensional dataset such as genomic data and brain imaging data, and the sample size is very small relative to the dimensionality. It is well-known that the sample covariance matrix performs poorly in high dimensions.
In the last decades, a fast growing literature on estimation of high-dimensional covariance matrix arises under structural assumptions or equivalent sparsity.  A common sparsity assumption in the literature is that all rows/columns of the covariance matrix lies in a sufficiently small $\ell_q$-ball around zero.  \cite{bickel2008covariance} proposed thresholding of the sample covariance matrix and \cite{rothman2009generalized} considered thresholding of the sample covariance matrix with more general thresholding functions.  \cite{cai2011adaptive} proposed an adaptive counterpart which achieves minimax optimality for sub-Gaussian variables with the $\ell_q$-ball sparsity assumption. Afterwards, more complicated covariance structures are considered, see for example \cite{cai2012adaptive,fan2013large,cai2016minimax,Fan2018Large}.  This research area is very active and the references listed here is only illustrative.

Compositional data arise in a wide range of applications. One typical type of compositional data is geochemical compositions of specimens such as rocks, sediments or soils.  The sum-to-one constraint makes the analysis of geochemical data difficult \citep{Chayes1960On} and the elements' distributions are typically skewed and it's often the case that there exist outliers or atypical observations \citep{Reimann2000Normal,Scealy2015Robust}.
  Another typical type of compositional data  is microbiome data and the current article is motivated by its metagenomic analysis. High-throughput sequencing techniques, such as targeted amplicon-based sequencing (TAS) and metagenomic profiling, provide large-scale genomic survey data of microbial communities in their natural habitats. However, these sequencing-based microbiome studies only provide us a relative measure of the abundances of community components rather than an absolute one. In fact, the microbial abundance is usually measured in read counts, which are
not directly comparable across samples due to the uneven total sequence counts of samples.
Therefore, the data are often normalized to relative abundances and sum to 1 for all microbes in a sample. In addition, the widespread outliers and high skewness have frequently been observed in sequencing samples \citep{Chen2017An,gao2019novel}.  The microbiome data fall into a class of high-dimensional leptokurtic and highly skewed compositional data with outliers that we focus on in this article.

 In metagenomic studies, it is of interest to understand the co-occurrence and co-exclusion relationship between human microbial taxa, which may shed light on the potential cause of complex diseases such as obesity, atherosclerosis, and Crohan's disease. Due to the unit-sum constraint of compositional data, conventional correlation analysis from the raw proportions fails to provide valid inference on the underlying biological mechanism. Thus it has been a long-standing question to model, estimate, and interpret the covariance structure for compositional data appropriately. As a pioneer work, \cite{aitchison1982statistical}  introduced several equivalent matrix specifications of covariance structures via the log-ratios of components. However, it's still unclear how to impose sparse structure in their models in high-dimensions due to a lack of direct covariances interpretation in these models. \cite{Friedman2012Inferring} focused on the correlations between latent variables based on log-ratio transformation of compositional data and proposed a method called SparCC under sparse assumption. \cite{Fang2015CCLasso} proposed a method called CCLasso based on least squares with $\ell_1$ penalty to infer the correlation network for latent variables of compositional data.
 \cite{Jiang2015Investigating} proposed a regularized estimation method for the basis covariance called REBACCA, which aims to  estimate the correlations between pairs of basis abundance with the log ratio transformation of metagenomic compositional data.
 \cite{cao2019large} introduced a COmposition-Adjusted Thresholding (COAT)  method to estimate the basis covariance matrix for high-dimensional compositional data, which has good interpretation for sparse structures.
The work of \cite{cao2019large} only derived the asymptotic convergence rate for data from a distribution with sub-Gaussian tails. The sub-Gaussianity assumption is an idealization of the complex random real world. Although the assumption facilitates the theoretical analysis, it is not realistic in practical applications as the collected modern data are often of
 low quality \citep{Qiang2018Adaptive}.  The existence of high skewness and outliers in microbiome data even makes the sub-Gaussianity assumptions seem more questionable.
 Figure \ref{fig:boxplot} shows the boxplots of estimation errors under matrix spectral norm  over 100 replications by  COAT, CCLasso, SparCC and REBACCA when synthetic data are generated from contaminated multivariate $t$ distribution. The detailed data generating setting is described in  Case 4  in Section \ref{sec:simulation}. From Figure  \ref{fig:boxplot}, we can see that the COAT, CCLasso, SparCC and REBACCA all perform unsatisfactorily when the underlying data are highly skewed and heavy-tailed, which is often the case for microbiome compositional data.
 Although there is  a lot of  literature on robust covariance matrix estimation in the presence of heavy-tailed data in high dimensions, such as \cite{xue2012regularized,liu2012high,He2017High,avella2018robust,Fan2018Large,He2018Variable,He2019Robust}, none of these work  considered the unit-sum constraint of compositional data.
Thus we are motivated to seek new robust procedures which can achieve the same minimax optimality when the data are high-dimensional, compositional, leptokurtic and highly skewed.


\begin{figure}[h]
  \centering
    \includegraphics[width=9cm, height=6.3cm]{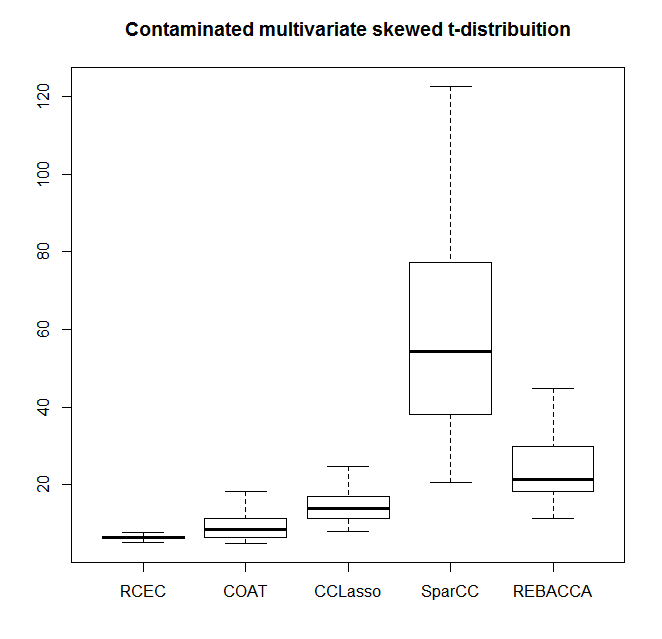}\\
 \caption{Boxplot of estimation errors under matrix spectral norm  over 100 replications by RCEC, COAT, CCLasso, SparCC and REBACCA for Case 4 in which synthetic data are generated from contaminated multivariate $t$ distribution, sample size $n=100$, dimensionality $p=100$.  }\label{fig:boxplot}
 \end{figure}

In this article, we assume the basis covariance matrix lies in a class of sparse covariance matrices $\cU(q,s_0(p),M)$ in (\ref{equ:classofcovariance}) and adopt the centered log-ratio covariance matrix as a proxy, which is approximately indistinguishable from the real basis covariance matrix in high-dimensions thanks to the unit-sum constraint of compositional data. We first construct a Median Of Means (MOM) estimator for the proxy matrix, which particularly fits to heavy-tailed data \citep{lerasle2011robust,Bubeck2013Bandits}. The final estimator is obtained by adaptively thresholding the MOM estimator for the centered log-ratio covariance matrix. We derive the same minimax convergence rate of the proposed estimator  as that in \cite{cai2011adaptive}, but we only assume finite fourth moments constraint.
 Simulation studies show that the proposed estimator outperforms some state-of-the-art estimators which ignore the heavy-tailedness and skewness of microbiome compositional data.
 Let's move back to Figure \ref{fig:boxplot}, the proposed method RCEC, abbreviated for \textbf{R}obust \textbf{C}ovariance \textbf{E}stimator for \textbf{C}ompositional data, outperforms the COAT, CCLasso, SparCC and REBACCA by a large margin in terms of estimation errors under the spectral matrix norm in the heavy-tailed and highly-skewed setting.   We also illustrate the method with a microbiome dataset, which helps us understand the heuristic dependence structure among bacteria taxa in the human gut.

We introduce the notation adopted throughout the paper. For any vector $\bmu=(\mu_1,\ldots,\mu_p)^\top \in \RR^p$, let $\|\bmu\|_2=(\sum_{i=1}^p\mu_i^2)^{1/2}$, $\|\bmu\|_\infty=\max_i|\mu_i|$. For a real number $a$, denote  $\lfloor a\rfloor$ as the largest integer smaller than or equal to $a$ and $(a)_{+}=\max\{a,0\}$. Let $I(\cdot)$ be the indicator function. For a matrix $\Ab=(a_{ij})$,  let $\Ab^\top$ be the transpose of $\Ab$, ${\rm Tr}(\Ab)$  the trace of $\Ab$, $\lambda_{\max}(\Ab)$ and $\lambda_{\min}(\Ab)$  the largest and smallest eigenvalue of a nonnegative definitive matrix $\Ab$ respectively and $\text{diag}(\Ab)$ be a vector composed of the diagonal elements of $\Ab$. Further note by $\|\Ab\|_1$, $\|\Ab\|_2$, $\|\Ab\|_F$ and $\|\Ab\|_{\max}$ the respective matrix $\ell_1$ norm, spectral norm, Frobenius norm and element-wise $\ell_\infty$ norm, i.e., $\|\Ab\|_1=\max_j\sum_i|a_{ij}|$, $\|\Ab\|_2=\sqrt{\lambda_{\max}(\Ab^\top\Ab)}$, $\|\Ab\|_F=\sqrt{\sum_{i,j}a_{ij}^2}$ and $\|\Ab\|_{\max}=\max_{i,j}|a_{ij}|$. Let $\Ab\succ 0$ denote that $\Ab$ is positive definite. For a set $\cH$, let $\text{Card}(\cH)$ be the cardinality of the set and $\one_p$ be a vector with all elements equal to 1 and $\zero$ be a vector with all elements equal to 0.

The rest of the paper is organized as follows. In Section 2 we introduce the class of sparse covariance matrices and review the basic relationship between the basis covariance matrix and the log-ratio covariance matrix. Section 3 introduces the robust covariance matrix estimator. In Section 4 we investigate the theoretical properties of the proposed estimator.  Section 5 presents the results of  thorough simulation studies. A real application to human gut microbiome data is given in Section 6. We discuss possible future research directions in Section 7 and all the detailed proofs of theorems are relegated to the Appendix.

\section{Preliminaries}
In this section we introduce some preliminary results on the compositional data analysis.
Let $\bZ=(Z_1,\ldots,Z_p)^\top$ with $Z_j>0$ for all $j$ be the latent basis variables. The observable composition variables $\bX=(X_1,\ldots,X_p)^\top$ are generated via normalizing the basis (latent) variables $\bZ$, i.e.,
\[
X_j=\frac{Z_j}{\sum_{i=1}^pZ_i}, \ \ \ j=1,\ldots,p.
\]

 It is infeasible to estimate the covariance of $\bZ$ owing to the apparent lack of identifiability. However, the basis covariance $\bOmega_0=(\omega_{ij}^0)$, defined as the covariance matrix of $Y_j=\log Z_j$, is approximately identifiable as long as it belongs to a class of large sparse covariance matrices \citep{cao2019large}.
In this article, the basis covariance matrix $\bOmega_0=(\omega_{ij}^0)$ is the parameter of interest. \cite{aitchison2003statistical} introduced the variation matrix $\Tb_0=(t_{ij}^0)$ defined by
\[
t_{ij}^0=\text{Var}(\log(X_i/X_j))=\text{Var}(\log Z_i-\log Z_j)=\text{Var}(Y_i-Y_j)=\omega_{ii}^0+\omega_{jj}^0-2\omega_{ij}^0,
\]
or in matrix form,
\[
\Tb_0=\bomega_0\one^\top+\one\bomega_0^\top-2\bOmega_0,
\]
where $\bomega_0=\text{diag}(\bOmega_0)$ and $\one=(1,\ldots,1)^\top$. The basis covariance matrix $\bOmega_0$ is unidentifiable from the above decomposition  as $\bomega_0\one^\top+\one\bomega_0^\top$ and $\bOmega_0$ are in general not orthogonal to each other.

The centered log-ratio covariance matrix $\bGamma_0=(\gamma_{ij}^0)$ is defined by
\[
\gamma_{ij}^0=\text{Cov}\Big\{\log(X_i/g(\bX)), \log(X_j/g(\bX))\Big\},
\]
where $g(\bx)=(\prod_{j=1}^px_j)^{1/p}$.

Thus for the variation matrix $\Tb_0$, we can similarly write
\[
\begin{split}
t_{ij}^0&=\text{Var}(\log(X_i)/\log(X_j))=\text{Var}\Big\{\log(X_i/g(\bX))-\log(X_j/g(\bX))\Big\}\\
&=\text{Var}\Big\{\log(X_i/g(\bX))\Big\}+\text{Var}\Big\{\log(X_j/g(\bX))\Big\}-2\text{Cov}\Big\{\log(X_i/g(\bX)), \log(X_j/g(\bX))\Big\}\\
&=\gamma_{ii}^0+\gamma_{jj}^0-2\gamma_{ij}^0,
\end{split}
\]
or in matrix form,
\begin{equation}\label{equ:T0Gamma0}
\Tb_0=\bgamma_0\one^\top+\one\bgamma_0^\top-2\bGamma_0,
\end{equation}
where $\bgamma_0=\text{diag}(\bGamma_0)$.

\begin{lemma}\label{lemma:approximation}
The components $\bgamma_0\one^\top+\one\bgamma_0^\top$ and $\bGamma_0$ in the decomposition (\ref{equ:T0Gamma0}) are orthogonal to each other. In addition, for the covariance matrices $\bOmega_0$ and $\bGamma_0$, we have
\[
\|\bOmega_0-\bGamma_0\|_{\max}\leq 3p^{-1}\|\bOmega_0\|_1.
\]

\end{lemma}

The proof of the lemma can be found in  \cite{cao2019large}, from which  we can conclude that the covariance matrix $\bOmega_0$ is approximately identifiable as long as $\|\bOmega_0\|_1=o(p)$. Assume that $\bOmega_0$ belongs to $\cU(q,s_0(p),M)$, the class of sparse covariances in \cite{bickel2008covariance},

\begin{equation}\label{equ:classofcovariance}
\cU(q,s_0(p),M)=\bigg\{\bOmega:\bOmega\succ0,\max_j\omega_{jj}\leq M,\max_i\sum_{j=1}^p|\omega_{ij}|^q\leq s_o(p)\bigg\},\ \ 0\leq q<1,
\end{equation}
then it can be shown that $\|\bOmega_0\|_1\leq M^{1-q}s_0(p)$. Thus $\bOmega_0$ and $\bGamma_0$ are asymptotically indistinguishable as long as $s_0(p)=o(p)$, which indicates $\bGamma_0$ can be used  as a good proxy for $\bOmega$. \cite{cao2019large} proposed a composition-adjusted thresholding (COAT) estimator based on this finding and obtained its convergence rate under the sub-Gaussian condition on $Y_j$'s. The sub-Gaussianity assumption can be too constrictive in practice, especially for microbiome data analysis, which motivates us to seek new procedures that can achieve the same minimax optimality when data are leptokurtic.

\section{Robust Covariance Matrix Estimator for Compositional Data}

In this section, we present the detailed robust covariance matrix estimation procedure for high-dimensional compositional data.  For notational simplicity, we let $W_i=\log(X_i/g(\bX))$ and $\bW=(W_1,\ldots,W_p)^\top$ and thus we have $\bGamma^0=\text{Cov}(\bW)$. Suppose that $(\bZ_k,\bX_k),k=1,\ldots,n$ are independent copies of $(\bZ,\bX)$. The compositions $\bX_k=(X_{k1},\ldots,X_{kp})^\top$ are observed while the latent bases $\bZ=(Z_{k1},\ldots,Z_{kp})^\top$ are unobservable. Notice that  $\bW$ are thus ``observed" by transforming the compositions $\bX$, and we denote $\bW_k=(W_{k,1},\ldots,W_{k,p})^\top$ with $W_{k,i}=\log(X_{ki}/g(\bX_k))$.

As $\bGamma_0$ acts as a proxy of $\bOmega$,
we first construct an estimate of $\bGamma_0$ and then apply adaptive thresholding to the estimate. From a robust perspective, we propose a medians of means estimator for $\bGamma_0$.

 Let $M\leq n$ be an integer and let $\cB=\{B_1,\ldots,B_M\}$ be a regular partition of $\{1,\ldots,n\}$, i.e.,
\[
\forall  K=1,\ldots,M, \ \ \big|\text{Card}{(B_K)}-\frac{n}{M}\big|\leq 1.
\]
The observations $\bW_k$ are partitioned into the $M$ blocks in $\cB$. Without loss of generality, we assume that $M$ is a factor of $n$, and $n=Md$. The samples in the $l$-th group is $\big\{\bW_{(l-1)d+1},\ldots,\bW_{(l-1)d+d}\big\}, l=1,\ldots,M$.
Then notice that $\gamma_{ij}=\text{Cov}(W_i,W_j)=E(W_iW_j)-E(W_i)E(W_j)$, the median of means estimator for $\gamma_{ij}$ can be constructed as:
\[
\hat\gamma_{ij}^M=\hat\mu_{ij}^M-\hat\mu_{i}^M\hat\mu_{j}^M, \ \ \text{with} \ \  \hat\mu_{ij}^M=\text{median}\big\{\overline{W}_{ij}^1,\ldots,\overline{W}_{ij}^M\big\}, \ \hat\mu_{i}^M=\text{median}\big\{\overline{W}_{i}^1,\ldots,\overline{W}_{i}^M\big\},
\]
where
\[
\overline{W}_{ij}^l=\frac{1}{d}\sum_{t=1}^dW_{(l-1)d+t,i}W_{(l-1)d+t,j}, \ \ \overline{W}_{i}^l=\frac{1}{d}\sum_{t=1}^dW_{(l-1)d+t,i}, \ \  l=1,\ldots,M, \ \ 1\leq i,j \leq p.
\]
Let $\hat\bGamma=(\hat\gamma_{ij}^M)$, and apply adaptive thresholding to $\hat\bGamma$. We obtain the robust estimator of $\bOmega$ as
\[
\hat\bOmega=(\hat{\omega}_{ij})_{p\times p} \ \ \text{with} \ \ \hat\omega_{ij}=\tau_{\lambda_{ij}}(\hat\gamma^M_{ij}),
\]
where $\lambda_{ij}>0$ are entry-wise thresholds and  $\tau_{\lambda}(\cdot)$ is a general thresholding function for which:
\begin{description}
  \item[(i)] $|\tau_{\lambda}(z)|\leq |y|$ for all $z$ and $y$ such that $|y-z|\leq \lambda$;
  \item[(ii)] $\tau_{\lambda}(z)=0$ for $|z|\leq \lambda$;
  \item[(iii)] $|\tau_{\lambda}(z)-z|\leq \lambda$ for all $z\in\RR$.
\end{description}

 The class of thresholding functions satisfying the three conditions include the soft thresholding rule $\tau_{\lambda}(z)=\text{sgn}(z)(|z|-\lambda)_{+}$, the adaptive lasso rule $\tau_{\lambda}(z)=z(1-|\lambda/z|^\eta)_{+}$ with $\eta\geq 1$, and the smoothly clipped absolute deviation thresholding rule proposed by \cite{rothman2009generalized}.

The performance of the robust estimator relies critically on the selected thresholds $\lambda_{ij}$. Similar to \cite{fan2013large} and \cite{avella2018robust}, we adopt the entry-dependent threshold

\begin{equation}\label{equ:threshold}
\lambda_{ij}=\lambda\left(\frac{\hat\gamma_{ii}^M\hat\gamma_{jj}^M\log p}{n}\right)^{1/2},
\end{equation}
where $\lambda>0$ is a constant. This is much simpler than the threshold used by \cite{cao2019large} as it does not require estimation of $\text{Var}\big\{(Y_i-EY_i)(Y_j-EY_j)\big\}$ and achieves the same optimality.

The thresholds in (\ref{equ:threshold}) depend on a tuning parameter $\lambda$ and can be selected by $V$-fold Cross Validation (CV).  In detail, denote by $\hat\bOmega^{(-v)}(\lambda)$ the robust estimate based on the samples excluding the $v$-th fold and $\hat \bGamma^{(v)}$ the robust median of means estimate based only on the samples in the $v$-th fold. The optimal value of $\lambda$ is chosen by minimizing the cross-validation error
\[
\lambda^*=\bargmin_{\lambda}\frac{1}{V}\sum_{t=1}^V\|\hat\bOmega^{(-v)}(\lambda)-\hat\bGamma^{(v)}\|_F^2.
\]

With the selected optimal tuning parameter $\lambda^*$, we then obtain the robust estimate based on the full dataset as the final estimate. The resulting estimate may not be positive-definite. To this end, we follow the approach in \cite{fan2013large} and choose $\lambda$ in the range where the minimum eigenvalue of the robust estimate is positive.
\section{Theoretical Analysis}
In this section we investigate the asymptotic properties of the robust estimator. Recall that $Y_j=\log Z_j$. Without loss of generality, we assume $EY_j=0$ for all $j$ throughout this section.  We assume the following conditions hold.
\vspace{1em}

\textbf{Assumption A}: Assume that $\max_{1\leq j\leq p}E(Y_j^4)=\kappa^2<\infty$.

\vspace{1em}

\textbf{Assumption B}: The basis covariance matrix $\bOmega_0$ belongs to the class
\[
\cU(q,s_0(p),M)=\bigg\{\bOmega:\bOmega\succ0,\max_j\omega_{jj}\leq M,\max_i\sum_{j=1}^p|\omega_{ij}|^q\leq s_0(p)\bigg\},\ \ 0\leq q<1,
\]
where $s_0(p)=O(p\sqrt{\log p/n})$, and $\log p=o(n)$.
\vspace{1em}

\textbf{Assumption C}: There exists a constant $\zeta>0$ such that $\min_i\omega_{ii}^0\geq \zeta$.

\vspace{1em}

\textbf{Assumption A}, \textbf{Assumption B} and \textbf{Assumption C} are common in the covariance matrix estimation literature, see, for example, \cite{cai2011adaptive,fan2013large,cao2019large}. \textbf{Assumption A} only requires that the fourth moments of $Y_j$ are uniformly bounded, which is much weaker than the sub-Gaussianity assumption in \cite{cao2019large}.  \textbf{Assumption B} imposed some conditions on the sparsity of the basis covariance matrix $\bOmega$ and the scaling between $p$ and $n$. The class of the sparse covariances are also considered in \cite{bickel2008covariance,cao2019large}. \textbf{Assumption C} is essential for adaptive thresholding methods.

The following theorem establishes the convergence rate of the median of means estimator $\hat\bGamma$ in terms of element-wise $\ell_\infty$-norm.

\begin{theorem}\label{theorem:1} Suppose that \textbf{Assumption A} and \textbf{Assumption B} hold, and let $\hat\bGamma$ be the median of means estimator based on the regular partitions in $\cB=\{B_1,\ldots,B_M\}$ with $M=\lceil(2+L)\log p\rceil$ for a positive constant $L$. Then we have for sufficiently large $n$ and a constant $C>0$,
\[
\Pr\left\{\|\hat\bGamma-\bOmega_0\|_{\max}\leq C\sqrt{\frac{\log p}{n}}\right\}\geq 1-\epsilon_{n,p},
\]
where $\epsilon_{n,p}\leq C_0p^{-L}$ for positive constants $C_0$ and $L$.

\end{theorem}

Theorem \ref{theorem:1} also provides a guidance for the selection of tuning parameter $M$. In fact, the choice of $M$ involves a compromise between bias and variance. For $M=1$ and $M=n$, it degenerates to sample mean and sample median, respectively. Sample mean is asymptotically unbiasd but does not concentrate exponentially fast in presence of heavy-tails, while the sample median concentrates exponentially fast but not to the population mean for asymmetric distributions. The choice $M=\lceil(2+L)\log p\rceil$ is an ideal one for which both goals are achieved simultaneously.

In the following theorem, we establish the convergence rate of the estimator $\hat\bOmega$ to $\bOmega_0$ in terms of matrix $\ell_2$-norm (spectral norm), which matches the minimax rate in \cite{cai2011adaptive}.
\begin{theorem}\label{theorem:2}
Suppose that \textbf{Assumption A}, \textbf{Assumption B} and \textbf{Assumption C} hold. Then there exists a positive constant $C$ such that
\[
\inf_{\bOmega\in\cU(q,s_0(p),M)}\Pr\left\{\|\hat\bOmega-\bOmega_0\|_2\leq Cs_0(p)\left(\frac{\log p}{n}\right)^{(1-q)/2}\right\}\geq 1-\epsilon_{n,p},
\]
where $\epsilon_{n,p}$ is a deterministic sequence that decreases to zero as $n,p\rightarrow\infty$.
\end{theorem}

Theorem \ref{theorem:2} generalizes Theorem 1 of \cite{cao2019large}. The lower bound of \cite{cao2019large} matches ours if the approximation error is dominated by the estimation error, i.e, $s_0(p)=O(p\sqrt{\log p/n})$.  This implies that our procedure is minimax optimal for a wider class of distributions containing  the sub-Gaussian distributions.

In the following theorem, we obtain the support recovery property of the estimator $\hat\bOmega$, where the support of $\bOmega_0$ refers to $\big\{(i,j),\omega_{ij}^0\neq 0\big\}$.
\begin{theorem}\label{theorem:3}
Suppose that \textbf{Assumption A}, \textbf{Assumption B} and \textbf{Assumption C} hold. Then the robust estimator $\hat\bOmega$ satisfies

\begin{equation}\label{equ:theorem31}
\Pr\left(\hat\omega_{ij}=0 \ \ \text{for all} \ \ (i,j) \ \ \text{with} \ \ \omega_{ij}^0=0\right)\rightarrow 1.
\end{equation}
Furthermore, if for a sufficiently large constant $C$,
\[
\bmin_{(i,j):\omega_{ij}^0\neq 0}|\omega_{ij}^0|/\sqrt{\omega_{ii}^0\omega_{jj}^0}\geq C\sqrt{\frac{\log p}{n}}
\]
then we have
\[
\Pr\big\{\text{sgn}(\hat\omega_{ij})=\text{sgn}(\omega_{ij}^0) \ \text{for all}\ (i,j)\big\}\rightarrow 1.
\]
\end{theorem}

Theorem \ref{theorem:3} illustrates that as long as the minimum signal is large enough, the proposed estimator can exactly  recover the support of $\bOmega_0$ with probability tending to 1.

\section{Simulation Study}\label{sec:simulation}
In this section, we conduct thorough numerical studies to investigate the empirical performance of the proposed estimator in various data-generating settings. We compare our Robust Covariance Estimator for Compositional (RCEC) data $\hat{\bOmega}_{rcec}$ with the oracle thresholding estimator $\hat{\bOmega}_{oracle}$, the COAT estimator $\hat\bOmega_{coat}$ in \cite{cao2019large}, the SparCC estimator in $\hat\bOmega_{sparcc}$ \cite{Friedman2012Inferring}, the CCLasso estimator $\hat\bOmega_{cclasso}$ in \cite{Fang2015CCLasso} and the REBACCA estimator $\hat\bOmega_{rebacca}$ in \cite{Jiang2015Investigating}. For the oracle thresholding estimator, we assume that the latent basis components are observable and apply the thresholding procedure to the median of means covariance matrix estimator.  In fact, $\hat{\bOmega}_{oracle}$ is the estimator that our method attempts to mimic.   For the implementation of COAT, we use the \textsf{R} code downloaded from \url{https://github.com/yuanpeicao/COAT}.
The tuning parameter $\lambda$ for the thresholding estimators $\hat{\bOmega}_{rcec},\hat{\bOmega}_{oracle},\hat\bOmega_{coat}$ was all chosen by 5-fold cross-validation with the soft thresholding rule $\tau_{\lambda}(z)=\text{sgn}(z)(|z|-\lambda)_{+}$ for a fair comparison. For the implementation of CCLasso and SparCC, we use the \textsf{R} code with its default parameter settings downloaded from \url{https://github.com/huayingfang/CCLasso}. For the implementation of REBACCA, we use the \textsf{R} code with its default parameter settings  downloaded from \url{http://faculty.wcas.northwestern.edu/~hji403/REBACCA.htm}.

To illustrate the robustness of the proposed method, we consider the following  data-generating settings. First, we consider the following structure for the covariance matrix $\bOmega_0$:

\vspace{0.5em}

Let $\bOmega_0=\text{diag}(\Ab_1,\Ab_2)$, where $\Ab_1=(\sigma_{ij})_{1\leq i,j\leq p/2}$, $\sigma_{ij}=\big(1-|i-j|/10\big)_{+}$, $\Ab_2=4\Ib_{p/2\times p/2}$, i.e,  $\bOmega_0$ is a two-block diagonal matrix, the first block is banded, and the second block is diagonal matrix with 4 along the diagonal.

\begin{table}[ht]
\caption{Simulation results for Case 1, the values in the parenthesis are the standard errors.}  
\bigskip
\centering 
\begin{tabular}{c c c c c c c c}  
\toprule[2pt]
   $p$  &  $\widehat{\Omega}_{rcec}$  & $\widehat{\Omega}_{oracle}$ & $\widehat{\Omega}_{coat}$ & & $\widehat{\Omega}_{cclasso}$           &$\widehat{\Omega}_{sparcc}$ &   $\widehat{\Omega}_{rebacca}$       \\ [0.5ex]
\hline  
          &   \multicolumn{7}{c}{\textbf{Matrix $L_1$ norm loss}}\\[0.5ex]
      50  & 8.139(0.390)  &5.822(0.878)  &7.881(0.405) &  &7.178(0.793) &17.099(1.711) &6.551(1.190)\\
     100  &7.904(0.359)  &6.995(0.537)  &7.305(0.476) &  &6.735(0.724) &32.255(2.678) &9.174(1.942)\\
     200  & 7.993(0.324) & 7.614(0.391) & 7.094(0.467) &  &7.046(0.568) & 62.209(3.181) & 15.754(2.396) \\
          &   \multicolumn{7}{c}{\textbf{Spectral norm loss}}\\[0.5ex]
      50  &6.556(0.285)  &4.385(0.737)  &6.184(0.315) &  &6.137(0.562) &10.138(0.792) &7.446(0.696)\\
     100  &6.560(0.309)  &5.538(0.419)  &5.789(0.313) &  &6.635(0.701) &13.870(0.933) &9.458(0.924)\\
     200  &6.643(0.281)  & 6.185(0.313) & 5.536(0.351)&  & 6.860(0.638)&19.473(0.988) &11.018(0.805) \\
          &    \multicolumn{7}{c}{\textbf{Matrix Fronbenius norm loss}}\\[0.5ex]
      50  &8.822(0.401)  &6.746(0.653)  &7.925(0.358) &  &7.112(0.633) &13.941(0.559) &8.611(0.790)\\
     100  &12.306(0.532) &10.766(0.558) &10.272(0.491)&  &9.143(0.694) &25.725(0.545) &13.977(0.755)\\
     200  & 17.506(0.528)& 16.410(0.592) & 13.941(0.588)&  & 12.853(0.653)& 50.453(0.688) &26.247(1.014)\\
          &    \multicolumn{7}{c}{\textbf{True positive rate}}\\[0.5ex]
      50  &0.623(0.045)  &0.782(0.062)  &0.746(0.043) &  &0.877(0.086) &1.000(0.000) &0.523(0.033)\\
     100  &0.618(0.037)  &0.718(0.033)  &0.735(0.027) &  &0.776(0.035) &1.000(0.000) &0.639(0.029)\\
     200  &0.621(0.025)  &0.671(0.026)  &0.753(0.024) &  &0.767(0.028)  &1.000(0.000) &0.684(0.021)\\
          &   \multicolumn{7}{c}{\textbf{False positive rate}}\\[0.5ex]
      50  &0.078(0.033)  &0.022(0.016)  &0.216(0.068) &  &0.394(0.133) &1.000(0.000) &0.031(0.009)\\
     100  &0.020(0.010)  &0.011(0.004)  &0.082(0.028) &  &0.107(0.029) &1.000(0.000) &0.025(0.003)\\
     200  &0.006(0.002)  &0.005(0.002)  &0.029(0.008) &  &0.036(0.008)  &1.000(0.000)  &0.026(0.003)\\
\bottomrule[2pt]
\end{tabular}
\label{table:Case1}  
\end{table}

The $(\bZ_k,\bX_k)$ for $k=1,\ldots,n$ are generated in the following way. We first generate $\bY_k$ in four different ways.

\vspace{0.5em}

\textbf{Case 1}: $\bY_k$ are independently drawn from multivariate normal distribution $\cN(\zero,\bOmega_0)$;

\vspace{0.5em}

\textbf{Case 2}: $\bY_k$ are independently drawn from multivariate $t$-distribution $t_\nu(\zero,\bOmega_0)$ with $\nu=3.5$, where the Probability Distribution Function (PDF) of a $p$-dimensional multivariate $t$ distribution $t_{\nu}(\bmu,\bSigma_{p\times p})$ is
	\begin{displaymath}
	\frac{{\Gamma\big((\nu+p)/2\big)}}{\Gamma(\nu/2)\nu^{p/2}\pi^{p/2}|\bSigma|^{1/2}}\bigg\{1+\frac{1}{\nu}(\bx-\bmu)^\top\bSigma^{-1}(\bx-\bmu)\bigg\}^{-(\nu+p)/2},
	\end{displaymath}
	 {where $\Gamma(\cdot)$ is the gamma function}.

\vspace{0.5em}

\textbf{Case 3}: $\bY_k$ are independently drawn from multivariate skewed $t$-distribution with four degrees of freedom and skew parameter equal to 20, i.e., $\mathcal{ST}_{p}(\bxi={\zero},\bOmega_0,\balpha={\bf 20},\nu=4)$, generated by function \texttt{rmvst} in \textsf{R} package \texttt{fMultivar}.

\vspace{0.5em}

\textbf{Case 4}:  $\bY_k$ are independently drawn from contaminated multivariate skewed $t$-distribution \cite{Azzalini2010The},  with 4 degree of freedom and skew parameter equal to 10. In detail, $\bY_k$ is generated as $\bY_k=(1-b_k)\bV_k^{(1)}+b_k\bV_k^{(2)}$, where $b_k\sim Binomial(1,0.05)$, $\bV_k^{(1)}\sim \mathcal{ST}_{p}(\bxi={\zero},\bOmega_0,\balpha={\bf 10},\nu=4)$ and $\bV_k^{(2)}\sim \cN(-8\one_p,\Ib)$.

\vspace{0.5em}

\begin{table}[ht]
\caption{Simulation results for Case 2, the values in the parenthesis are the standard errors.}  
\bigskip
\centering 
\scalebox{0.9}{
\begin{tabular}{c c c c c c c c}  
\toprule[2pt]
   $p$  &  $\widehat{\Omega}_{rcec}$  & $\widehat{\Omega}_{oracle}$ & $\widehat{\Omega}_{coat}$ & & $\widehat{\Omega}_{cclasso}$           &$\widehat{\Omega}_{sparcc}$ &   $\widehat{\Omega}_{rebacca}$       \\ [0.5ex]
\hline  
          &   \multicolumn{7}{c}{\textbf{Matrix $L_1$ norm loss}}\\[0.5ex]
      50  &8.172(0.719)  &6.951(1.195)  &14.343(9.002) &  &21.208(20.966) &122.247(112.76)  &35.905(30.64)\\
     100  &8.084(0.679)  &7.606(0.788)  &15.682(10.673)&  &24.207(23.803) &252.876(254.277) &54.368(37.343)\\
     200  &8.277(0.635)  &8.205(0.949)  &16.036(13.363)&  &27.744(41.712) &455.974(487.957) &101.143(136.503)\\
          &   \multicolumn{7}{c}{\textbf{Spectral norm loss}}\\[0.5ex]
      50  &6.601(1.186)  &5.965(1.670)  &13.119(9.309) &  &20.374(13.346) &58.937(50.965)   &27.525(17.450)\\
     100  &6.695(1.173)  &6.160(1.252)  &14.542(10.898)&  &20.357(11.889) &105.334(111.803) &33.370(18.152)\\
     200  &6.914(1.148)  &6.849(1.426)  &14.867(13.692)&  &20.564(15.549) &168.714(180.732) &46.897(40.224)\\
          &    \multicolumn{7}{c}{\textbf{Matrix Fronbenius norm loss}}\\[0.5ex]
      50  &16.403(3.518) &15.943(3.815) &31.805(15.447)&  &36.447(18.930) &71.412(48.314)   &44.086(22.683)\\
     100  &21.028(4.690) &20.468(4.870) &43.936(22.854)&  &47.361(24.534) &134.285(106.350) &68.217(31.898)\\
     200  &29.848(5.615) &29.479(5.934) &61.298(32.621)&  &64.495(34.214) &239.031(178.771) &117.097(65.765)\\
          &    \multicolumn{7}{c}{\textbf{True positive rate}}\\[0.5ex]
      50  &0.495(0.046)  &0.656(0.077)  &0.603(0.089)  &  &0.461(0.335)   &1.000(0.000)     &0.262(0.109)\\
     100  &0.504(0.042)  &0.588(0.052)  &0.601(0.096)  &  &0.431(0.239)   &1.000(0.000)     &0.389(0.133)\\
     200  &0.500(0.033)  &0.537(0.039)  &0.615(0.081)  &  &0.434(0.212)   &1.000(0.000)     &0.460(0.104)\\
          &   \multicolumn{7}{c}{\textbf{False positive rate}}\\[0.5ex]
      50  &0.030(0.019)  &0.012(0.007)  &0.114(0.051)  &  &0.180(0.121)   &1.000(0.000)     &0.027(0.009)\\
     100  &0.008(0.005)  &0.005(0.003)  &0.049(0.020)  &  &0.051(0.026)   &1.000(0.000)     &0.024(0.006)\\
     200  &0.003(0.001)  &0.002(0.001)  &0.024(0.008)  &  &0.017(0.009)   &1.000(0.000)     &0.023(0.004)\\
\bottomrule[2pt]
\end{tabular}}
\label{table:Case2}  
\end{table}

Then $\bZ_k=(Z_{k1},\ldots,Z_{kp})^\top$ and $\bX_k=(X_{k1},\ldots,X_{kp})^\top$ were obtained by the transformations
\[
Z_{kj}=\exp(Y_{kj}), \ \ \text{and}\ \ X_{kj}=\frac{Z_{kj}}{\sum_{i=1}^pZ_{ki}}, \ \ j=1,\ldots, p.
\]
In \textbf{Case 1},  the latent variables $\bY_k$ are generated from Gaussian distribution. In \textbf{Case 2},  the latent variables $\bY_k$ are generated from symmetric multivariate $t$ distribution with degree 3.5.  In \textbf{Case 3},  the latent variables $\bY_k$ are generated from skewed $t$-distribution.  \textbf{Case 4} is from \cite{avella2018robust}, in which the latent variables $\bY_k$ are generated from contaminated skewed $t$-distribution with four degrees of freedom and skew parameter equal to 10.  We set the sample size $n=100$ and the dimension $p=50,100,200$, and conducted 200 replications for each setting. To evaluate the performance of different estimators, we adopt matrix $L_1$-norm, spectral norm, and Frobenius norm to measure the estimation losses and  use  the true positive rate and false positive rate to assess the quality of support
recovery. In all simulation settings, we let $M=\lceil(2+L)\log p\rceil$ with $L=1$.

The simulation results for Cases 1 are presented in Tables \ref{table:Case1}. From Table \ref{table:Case1}, we can see that the proposed RCEC performs comparably with the COAT and CCLasso method, while performs better than  SparCC and REBACCA,  in terms of both estimation losses and support recovery, when the underlying variables are from Gaussian distribution. In addition, we can also see that the RCEC method performs almost the same with the estimator  $\widehat{\Omega}_{oracle}$, which indicates that the RCEC method can act as if the latent variables  generating the compositions are observed. It seems that CCLasso method performs the best in Case 1, with smaller estimation error losses and higher True Positive Rate (TPR).

The simulation results for Cases 2 are presented in Tables \ref{table:Case2}. From Table \ref{table:Case2} we can see that, the proposed RCEC outperforms all its competitors in terms of  estimation losses  by a large margin, and performs almost the same with the estimator  $\widehat{\Omega}_{oracle}$, which shows the robustness and superiority of the RCEC method when the underlying variables are from heavy-tailed distributions. {As for the support recovery, the proposed RCEC method seems to have satisfactory true positive rates and lower false positive rates, compared with those by COAT, CCLasso, SparCC and REBACCA methods.} We can also see that the proposed RCEC method performs comparably with the estimator $\widehat{\Omega}_{oracle}$ in terms of both estimation losses and support recovery. The same conclusions for Case 2 can be drawn for Case 3 and Case 4,  based on the results presented in Table \ref{table:Case3} and Table \ref{table:Case4} in the appendix.

In conclusion, the proposed RCEC performs well in various data generating scenarios in terms of both estimation losses and support recovery, while the COAT method proposed by \cite{cao2019large} no longer works well when the underlying variables are from heavy-tailed or highly skewed asymmetric distributions. In other word, the RCEC may be used as an alternative of the COAT, CCLasso, SparCC and REBACCA in covariance matrix estimation for compositional data.

\section{Real data example}
In this section, we illustrate our estimator with a microbiome dataset in human gut. It is well-known that the gut microbiome plays a critical role in energy extraction from the diets. The microbiome taxa interacts with the immune system and thus has a profound influence on human health. The interactions among microbial taxa may provide new insight into the cause of disease such as obesity. We apply the proposed method to analyze a human gut microbiome dataset. The dataset was also analyzed in \cite{Coyte2015The,cao2019large}, from which one can get the detailed description of the dataset. In this real data example, we are also interested in investigating the underlying correlation structures among bacterial genera between lean and obese subjects as in \cite{cao2019large}. The dataset was divided into a lean group and an obese group according to the BMI index. A subject is assigned to the lean group if its $\text{BMI}<25$
and assigned to the obese group otherwise. It turns out the lean group has 63 subjects and the obese group has 35 subjects. We focused
on the 40 bacterial genera which appeared in at least four samples in each group. The original data were transformed into compositions after the zero counts were replaced by 0.5.  Figure \ref{fig:kur-hist} shows the frequency histogram of the sample kurtosis for the 40 bacterial genera in the obese group and in the lean group. About one half of the 40 bacterial genera show larger sample kurtosis than the value 9, which is the theoretical kurtosis of $t_5$ distribution. Thus it is more reasonable to take the heavy-tailed property into consideration.

\begin{figure}[h]
  \centering
  \begin{minipage}[!t]{0.35\linewidth}
    \includegraphics[width=6cm, height=6cm]{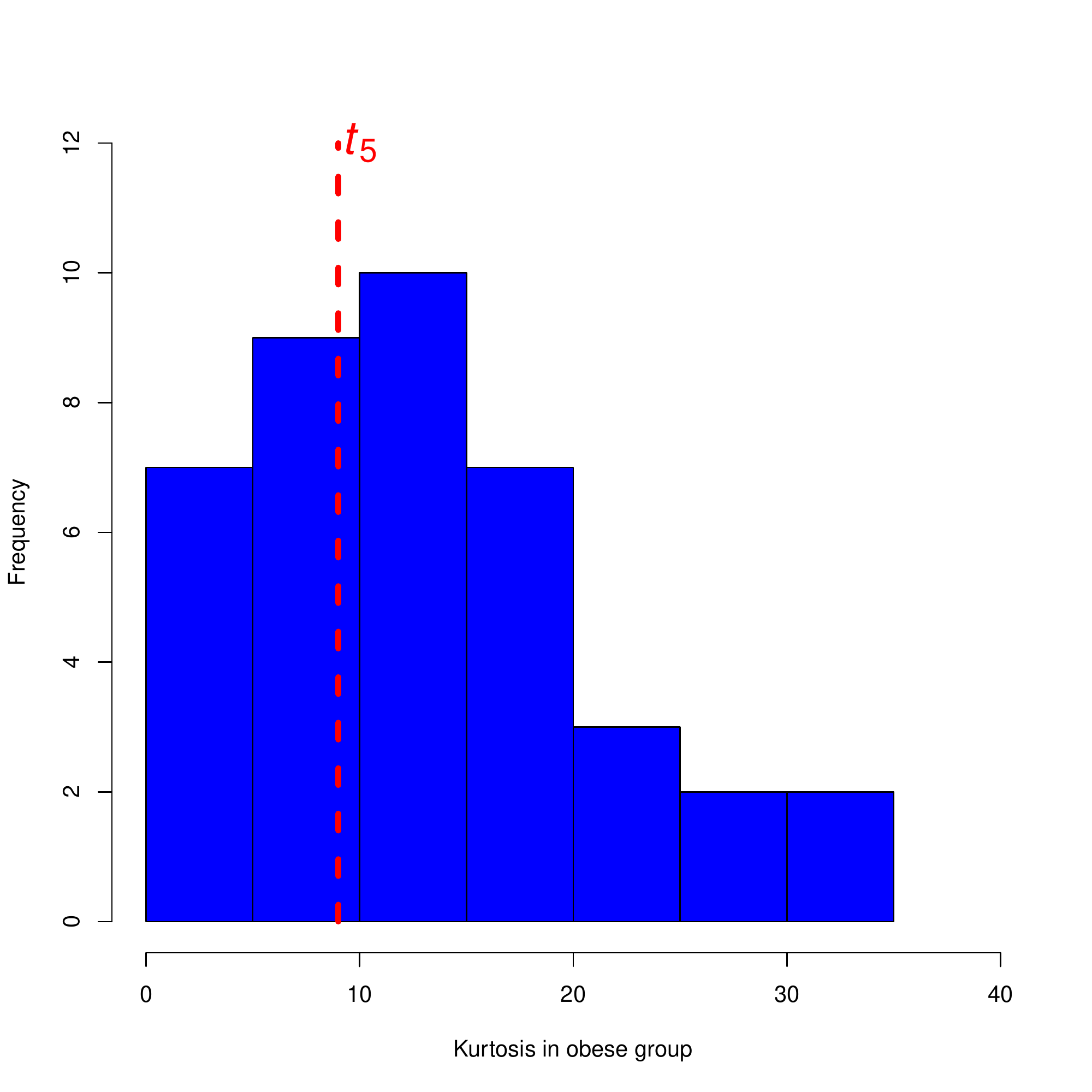}\\
  \end{minipage}
  \begin{minipage}[!t]{0.35\linewidth}
    \includegraphics[width=6cm, height=6cm]{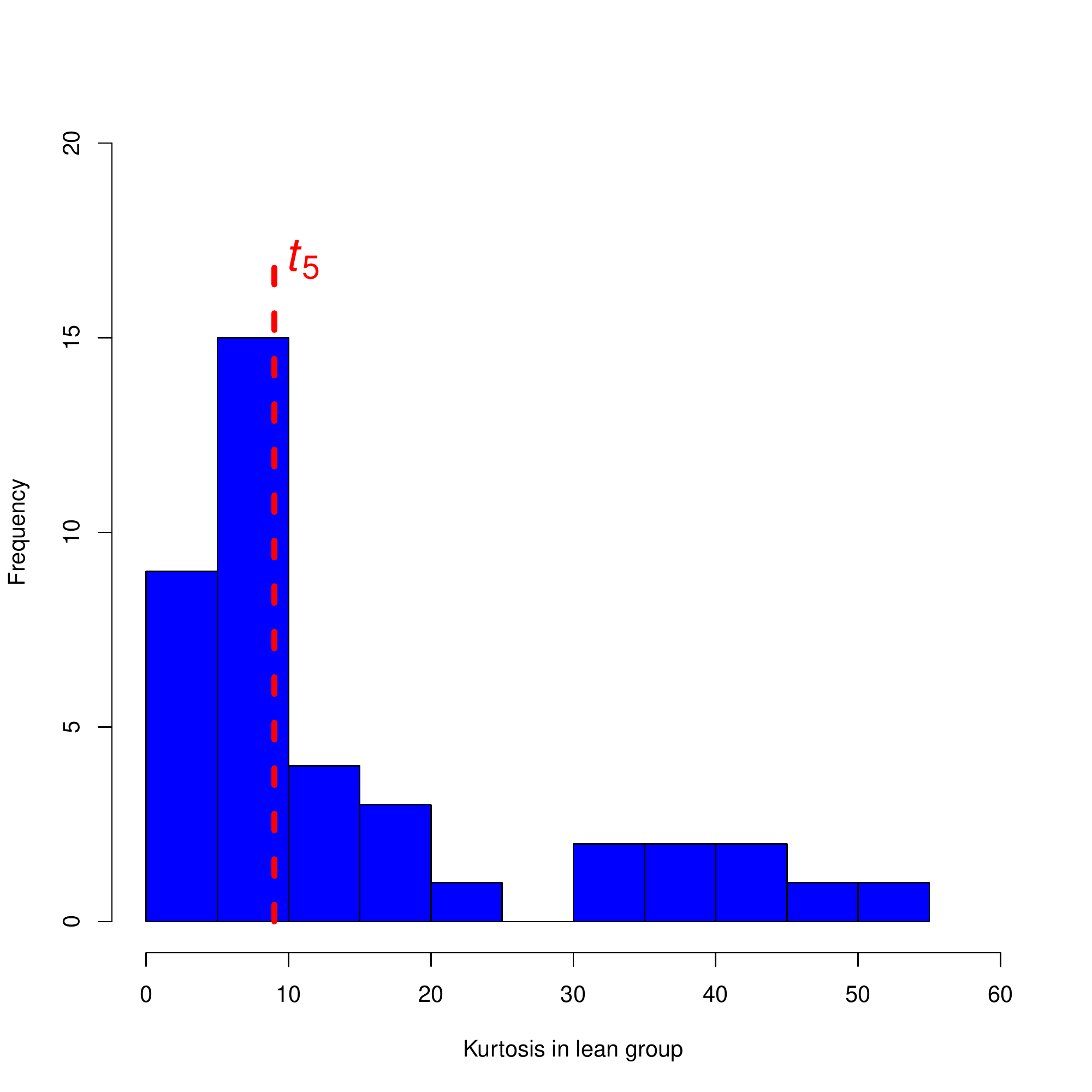}\\
  \end{minipage}
 \caption{Histogram of the sample kurtosis for the 40 bacterial genera in the obese group (left)  and lean group (right) and the red dashed line is the theoretical kurtosis of $t_5$ distribution. }\label{fig:kur-hist}
 \end{figure}

\begin{figure}[!h]
  \centering
  \includegraphics[width=15cm, height=10cm]{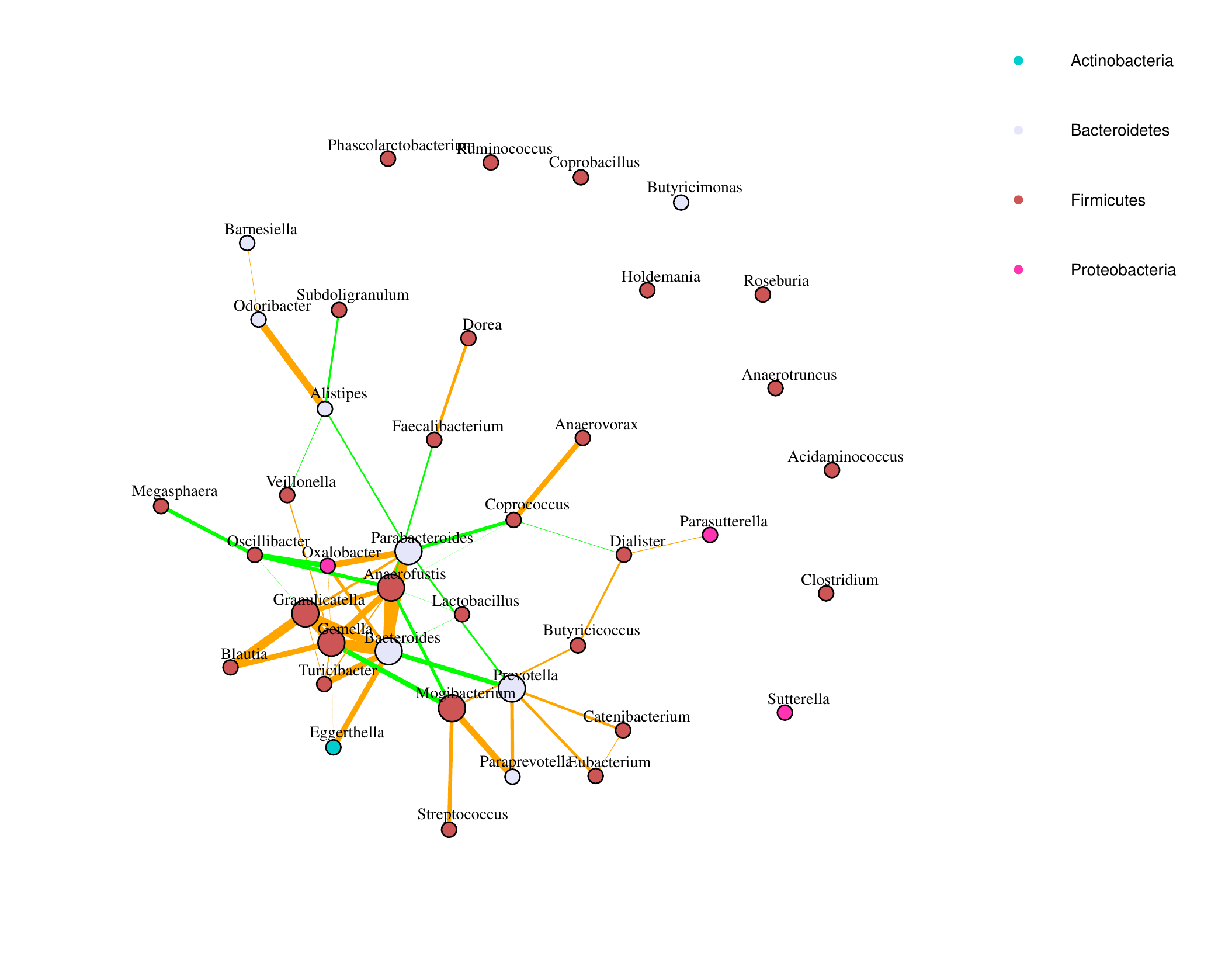}
  \caption{Correlation network identified by the RCEC method for the obese group. Positive correlations are displayed in orange and negative correlations are displayed in green. The size of the nodes indicates the magnitude of degree and the thickness of edges indicates the magnitude of correlations.}\label{fig:realdataobese}
\end{figure}

\begin{table}[!h]
\caption{ Numbers of positive and negative correlations and stability of correlation networks. }  
\renewcommand{\arraystretch}{1.2}
\bigskip
\centering 
\begin{tabular}{ccccccccccc}  
\toprule[2pt]
     & \multicolumn{2}{c}{\textbf{Lean}}  & \multicolumn{3}{c}{\textbf{Obese}} \\
\cmidrule(lr){2-3}  \cmidrule(lr){5-6}
         &  $\widehat{\Omega}_{rcec}$    &  $\widehat{\Omega}_{coat}$     & & $\widehat{\Omega}_{rcec}$           &    $\widehat{\Omega}_{coat}$         \\
\hline  
    \text{Positive Correlations}     &120    &47 &  &34      &79\\
    \text{Negative Correlations}     &154    &61 &  &17      &103\\
    \text{Network stability}         &0.712    &0.584 &  &0.596      &0.563\\

\bottomrule[2pt] 
\end{tabular}
\label{table:realdata}  
\end{table}

\begin{figure}[!h]
  \centering
  \includegraphics[width=15cm, height=9cm]{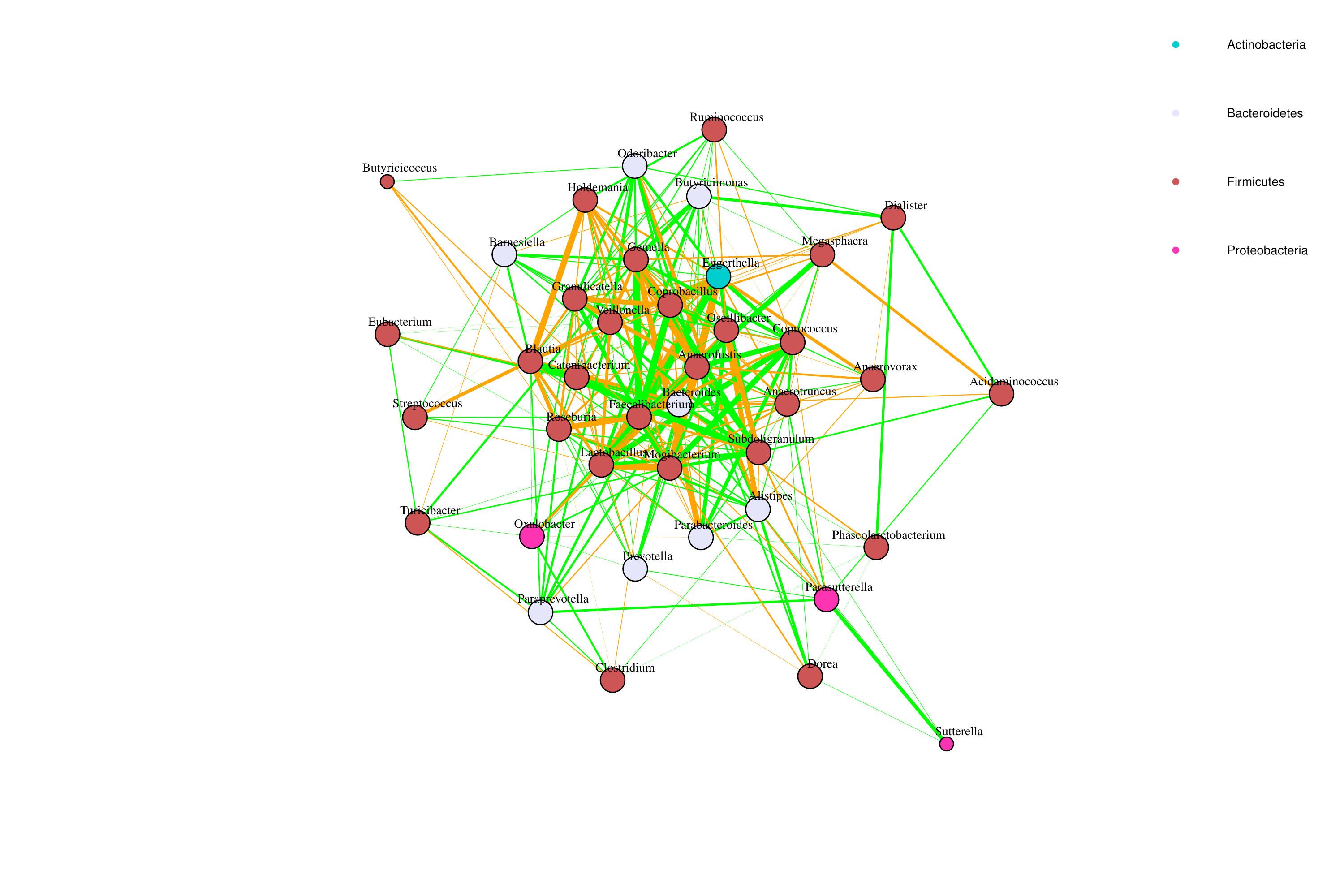}
  \caption{Correlation network identified by the RCEC method for the lean group. Positive correlations are displayed in orange and negative correlations are displayed in green. The size of the nodes indicates the magnitude of degree and the thickness of edges indicates the magnitude of correlations. }\label{fig:realdatalean}
\end{figure}

We applied the RCEC method and the COAT method with the soft thresholding rule to the obese and lean groups, and the tuning parameter $\lambda$ is selected by 5-fold cross validation. The identified edges are denoted as $E_{obese}$ and $E_{lean}$, respectively. We construct a network among the bacterial genera according to the estimated correlation matrix, in which an edge $(i,j)$ represents the correlation between bacteria $i$ and $j$. The stability of support recovery is assessed by the following strategy. We generate 35 bootstrap samples with replacement from the obese group  and 63 bootstrap samples with replacement for the lean group and perform RCEC and COAT procedures on the bootstrapped samples for each group. We repeat the above bootstrap procedure for 100 times. For each edge in  $E_{obese}$ (or $E_{lean}$), we count the times of its occurrences in 100 bootstrap replicates.  The stability of support recovery  is computed as these edges' average proportions of occurrences.     Finally, we  retain the edges in the network  identified by at least 50 bootstrap replicates. Table \ref{table:realdata} displays the numbers of positive and negative correlations and the network stability for the obese and lean groups. We see that the RCEC method achieves higher network stability than the COAT method. There is a bit discrepancy between the results of the COAT method and those derived from \cite{cao2019large}. This may be due to the randomness in cross-validation, the randomness of bootstrapping and in addition, we retained edges in the network  identified by at least 50 bootstrap replicates rather than 80 in their original paper.

The correlation networks identified by the RCEC method for the obese group and lean group  are displayed in Figure  \ref{fig:realdataobese} and Figure \ref{fig:realdatalean}, respectively. From  Figure  \ref{fig:realdataobese} and Figure \ref{fig:realdatalean}, we can clearly see that the correlation networks for the obese and lean groups differ significantly. It can be seen that the lean group shows more complex interactions than the obese group, which further indicates that the obese microbiome is less modular. The finding has been illustrated in literatures and the reason may be the adaption of the microbiome to low-diversity environments \citep{Sharon2012Metagenomic}. In addition, we can see that in the obese group, the gut microbial correlation network tends to have more competitive (or equivalently negative) interactions than cooperative (or equivalently positive) interactions, while the opposite happens in the lean group. The finding is different from that by the COAT method, in which the gut microbial correlation network tends to have more competitive interactions than cooperative ones for both the obese and lean group (see Table \ref{table:realdata}). Maybe the decrease of competitive interactions among the bacterial genera is closely related to the cause of obese and needs to further validated.

\section{Discussion}
A robust covariance matrix estimator for high-dimensional compositional data is proposed in this article, which is shown to enjoy minimax convergence rate in a large class of sparse covariance matrices. In essence, the robustness is achieved by the median of means estimator for the centered log-ratio covariance matrix, which concentrates exponentially fast only under bounded fourth moment condition. Another promising robust estimator would be the Huber's M-estimator \citep{Huber1964Robust}. For i.i.d. copies $V_1,\ldots,V_n$ of a real random variable $V$ with mean $\mu$, Huber's M-estimator of $\mu$ satisfies
$\sum_{i=1}^n\Phi_K(Z_i-\mu)=0$, where $\Phi_K(x)=\min\big\{K,\max(-K,x)\big\}$ is the Huber function. Thus an estimator can be similarly constructed by the Huber's M-estimators for means $E(W_iW_j)$, $EW_i$, $EW_j$. The truncation $K$ is a parameter that trades off bias and robustness and should be carefully dealt with \cite{fan2017estimation,avella2018robust}. We leave this as a future work.


\vspace{2em}

\begin{center}
APPENDIX: PROOFS OF MAIN THEOREMS AND ADDITIONAL SIMULATION RESULTS
\end{center}

\begin{appendices}
\section{Proofs of Main Theorems}
We first present two useful lemmas. Lemma \ref{lemma:LO2011}  is a simplified version of Proposition 1 in \cite{lerasle2011robust}. Lemma \ref{lemma:MOMconcentration} gives the convergence rate of the medians of means estimator $\hat\bGamma$ to $\bGamma$ in terms of element-wise $\ell_\infty$-norm.
\begin{lemma}\label{lemma:LO2011}
Let $Z_1,\ldots,Z_n$ be independent identically distributed random variables with $EZ_1=\bmu$ and $\text{Var}(Z_1)=\sigma^2$. Let $\delta\in(0,1)$ and $M\leq n/2$, and let $\cB=\{B_1,\ldots,B_M\}$ be a regular partition of $\{1,\ldots,n\}$ and $\hat\mu^M$  the median of means estimator of $\bmu$ based on blocks $\cB$. Then if $M\geq \log (\delta^{-1})$, we have that for some constant $K\leq 2(6e)^{1/2}$,
\[
\Pr\left\{\hat\mu^M-\mu\geq K\left(\frac{\sigma^2\log \delta^{-1}}{n}\right)^{1/2}\right\}\leq \delta.
\]
\end{lemma}

\begin{lemma}\label{lemma:MOMconcentration}
For the median of mean estimator $\hat\bGamma=(\hat\gamma^M_{ij})$, we have that for some constant $C>0$,
\[
\Pr\left\{\|\hat\bGamma-\bGamma^0\|_{\max}\geq C\sqrt{\frac{\log p}{n}}\right\}\leq 2p^{-L}(1+p^{-1}).
\]

\end{lemma}
\begin{proof}
Recall that $\hat\gamma^M_{ij}=\hat\mu^M_{ij}-\hat\mu_{i}^M\hat\mu_{j}^M$. By Lemma \ref{lemma:LO2011}, we have that for all $i,j\in\{1,\ldots,p\}$,
\[
\Pr\left\{\bmax_i|\hat\mu_{i}^M-EW_i|\geq K\sqrt{\frac{\gamma_{ii}^0\log p}{n}}\right\}\leq 2p^{-(1+L)},
\]
and
\[
\Pr\left\{\bmax_{i,j}|\hat\mu_{ij}^M-E(W_iW_j)|\geq K\sqrt{\frac{\text{Var}(W_iW_j)\log p}{n}}\right\}\leq 2p^{-L}.
\]
Thus with probability at least $1-2p^{-(1+L)}$,  we have
\[
\begin{split}
\bmax_{i,j}|\hat\mu_{i}^M\hat\mu_{j}^M-EW_iEW_j|&\leq\bmax_{i,j}|\hat\mu_{i}^M\hat\mu_{j}^M-EW_i\hat\mu_{j}^M|
+\bmax_{i,j}|EW_i\hat\mu_{j}^M-EW_iEW_j|\\
&\leq \bmax_{i,j}|\hat\mu_{i}^M-EW_i|\Big(|\hat\mu_{j}^M-EW_j|+|EW_j|\Big)+\bmax_{i,j}|EW_i||\hat\mu_{j}^M-EW_j|\\
&\leq 2K\bmax_{i,j}|EW_i|\sqrt{\frac{\gamma_{jj}^0\log p}{n}}+K^2\bmax_{i}\frac{\gamma_{ii}^0\log p}{n}.
\end{split}
\]
Furthermore,
\[
\begin{split}
\Pr\left\{\bmax_{ij}|\hat\gamma^M_{ij}-\gamma_{ij}^0|\geq C\sqrt{\frac{\log p}{n}}\right\}=&\Pr\left\{\bmax_{ij}\Big|(\hat\mu^M_{ij}-E(W_iW_j))+(EW_iEW_j-\hat\mu_{i}^M\hat\mu_{j}^M)\Big|\geq C\sqrt{\frac{\log p}{n}}\right\} \\
\leq &\Pr\left\{\bmax_{ij}\Big|(\hat\mu^M_{ij}-E(W_iW_j))\Big|+\bmax_{i,j}\Big|(EW_iEW_j-\hat\mu_{i}^M\hat\mu_{j}^M)\Big|\geq C\sqrt{\frac{\log p}{n}}\right\}\\
\leq &\Pr\left\{\bmax_{ij}\Big|(\hat\mu^M_{ij}-E(W_iW_j))\Big|\geq \frac{C}{2}\sqrt{\frac{\log p}{n}}\right\}\\
&+\Pr\left\{\bmax_{i,j}\Big|(EW_iEW_j-\hat\mu_{i}^M\hat\mu_{j}^M)\Big|\geq C\sqrt{\frac{\log p}{n}}\right\}\\
\leq &2p^{-L}(1+p^{-1}),
\end{split}
\]
which concludes the result in Lemma \ref{lemma:MOMconcentration}.
\end{proof}

\section*{Proof of Theorem \ref{theorem:1}}
\begin{proof}
By \textbf{Assumption B}, we have that
\[
\|\bOmega_0\|_1=\bmax_i\sum_{j=1}^p|\omega_{ij}^0|=\bmax_i\sum_{j=1}^p|\omega_{ij}^0|^{1-q}|\omega_{ij}^0|^q
\leq\bmax_i\sum_{j=1}^p (\omega_{ii}^0\omega_{jj}^0)^{(1-q)/2}|\omega_{ij}^0|^q \leq M^{1-q}s_0(p).
\]
Thus further by Lemma \ref{lemma:approximation}, there exists a constant $C_1$ such that
\[
\|\bOmega_0-\bGamma_0\|_{\max}\leq 3p^{-1}M^{1-q}s_0(p)\leq C_1\sqrt{\frac{\log p}{n}}.
\]
Thus by Lemma \ref{lemma:MOMconcentration}, we have that for sufficiently large $C$,
\[
\begin{split}
\Pr\left\{\|\hat\bGamma-\bOmega_0\|_{\max}\geq C\sqrt{\frac{\log p}{n}}\right\}\leq &\Pr\left\{\|\hat\bGamma-\bGamma_0\|_{\max}+\|\bOmega_0-\bGamma_0\|_{\max}\geq C\sqrt{\frac{\log p}{n}}\right\}\\
\leq & \Pr\left\{\|\hat\bGamma-\bGamma_0\|_{\max}\geq C/2\sqrt{\frac{\log p}{n}}\right\}\\
\leq & 2p^{-L}(1+p^{-1}),
\end{split}
\]
where the last inequality is by Lemma \ref{lemma:MOMconcentration}, which concludes the theorem.
\end{proof}
\section*{Proof of Theorem \ref{theorem:2}}
\begin{proof}
Define two events as
\[
E_1=\left\{|\hat\gamma^M_{ij}-\omega^0_{ij}|\leq \lambda_{ij} \ \ \text{for all}\ \ i,j\right\},  \ \ E_2=\left\{\hat\gamma^M_{ii}\hat\gamma^M_{jj}\leq 2\omega_{ii}^0\omega_{jj}^0\ \ \text{for all}\ \ i,j\right\}.
\]
 We divide the proof of  Theorem \ref{theorem:2} into two steps. In the first step, we show that on the event $E_1\bigcap E_2$, we have
 \[
 \|\hat\bOmega-\bOmega_0\|_2\leq C_0s_0(p)\left(\frac{\log p}{n}\right)^{(1-q)/2},
 \]
where $C_0$ is a constant depending on $q$ and $\lambda$. In the second step, we show that for a positive deterministic sequence $\epsilon_{n,p}$ which converge to zero when $(\log p)/n\rightarrow0$, we have that $\Pr\{E_1\bigcap E_2\}\geq 1-\epsilon_{n,p}$.

\textbf{Step I:} First by the properties of the thresholding functions $\tau_{\lambda}(\cdot)$, we have
\[
\begin{split}
\sum_{j=1}^p|\tau_{\lambda_{ij}}(\hat{\gamma}_{ij}^M)-\omega_{ij}^0|=&\sum_{j=1}^p|\tau_{\lambda_{ij}}(\hat{\gamma}_{ij}^M)-\omega_{ij}^0|I(|\hat\gamma_{ij}^M|\geq \lambda_{ij})+\sum_{j=1}^p|\omega_{ij}^0|I(|\hat{\gamma}_{ij}^M|<\lambda_{ij})\\
\leq&\underbrace{2\sum_{j=1}^p\lambda_{ij}I(|\omega_{ij}^0|\geq \lambda_{ij})}_{L_1}+\underbrace{\sum_{j=1}^p|\tau_{\lambda_{ij}}(\hat{\gamma}_{ij}^M)-\omega_{ij}^0|I(|\hat\gamma^M|\geq\lambda_{ij})}_{L_2}\\
&+\underbrace{\sum_{j=1}^q|\omega_{ij}^0|I(|\hat\gamma^M_{ij}|<\lambda_{ij})}_{L_3}.\\
\end{split}
\]
On the event $E_1$, by property (i) of the thresholding function, we have
\[
L_2\leq 2\sum_{j=1}^p|\omega_{ij}^0|I(|\omega_{ij}^0|<\lambda_{ij}),
\]
and by triangular inequality, we have that
\[
L_3\leq \sum_{j=1}^q|\omega_{ij}^0|I(|\omega_{ij}^0|<2\lambda_{ij}).
\]
Combining the above inequalities, on the event $E_2$, we have
\[
\begin{split}
\sum_{j=1}^p|\tau_{\lambda_{ij}}(\hat{\gamma}_{ij}^M)-\omega_{ij}^0|\leq & (4+2^{1-q})\sum_{j=1}^p\lambda_{ij}^{1-q}|\omega_{ij}^0|^q\\
\leq & C_0s_0(p)\left(\frac{\log p}{n}\right)^{(1-q)/2}.
\end{split}
\]
Further by $\|\Ab\|\leq \|\Ab\|_1$ for any symmetric matrix $\Ab$,  it only remains to show the result in Step II.
\end{proof}

\textbf{Step II:}  By \textbf{Assumption C}, there exists $\zeta>0$ such that $\min_{i}\omega_{ii}>\zeta$. Notice that

\begin{equation}\label{equ:S222}
\hat\gamma^M_{ii}\hat\gamma^M_{jj}=\omega_{ii}^0\omega_{jj}^0+(\hat\gamma^M_{ii}-\omega_{ii}^0)\hat\gamma^M_{jj}
+(\hat\gamma^M_{jj}-\omega_{jj}^0)\hat\gamma^M_{ii}-(\hat\gamma^M_{ii}-\omega_{ii}^0)(\hat\gamma^M_{jj}-\omega_{jj}^0).
\end{equation}

In Theorem \ref{theorem:1}, we have obtained that
\[
\Pr\left\{\|\hat\bGamma-\bOmega_0\|_{\max}\leq C\sqrt{\frac{\log p}{n}}\right\}\geq 1-\epsilon_{n,p},
\]
therefore, for large enough $n$,
\[
\Pr\left\{\hat \gamma_{ii}^M\hat\gamma_{jj}^M\geq \frac{\zeta^2}{2} \ \ \text{for all} \ \ i,j\in\{1,\ldots,p\}\right\}\geq 1-\frac{\epsilon_{n,p}}{4}.
\]
Then we have

\begin{equation}\label{equ:S1}
\begin{split}
\Pr\left\{\bmax_{i,j}\frac{|\hat{\gamma}_{ij}^M-\omega_{ij}^0|}{\sqrt{\hat\gamma_{ii}^M\hat\gamma_{jj}^M}}\geq \lambda\sqrt{\frac{\log p}{n}}\right\}
\leq& \Pr\left\{\bmax_{i,j}{|\hat{\gamma}_{ij}^M-\omega_{ij}^0|} \geq \lambda\sqrt{\frac{\bmin_{ij}(\hat\gamma_{ii}^M\hat\gamma_{jj}^M)\log p}{n}}\right\}\\
\leq& \Pr\left\{\bmax_{i,j}{|\hat{\gamma}_{ij}^M-\omega_{ij}^0|} \geq \lambda\sqrt{\frac{\zeta^2\log p}{2n}}\right\}+\frac{\epsilon_{n,p}}{4}\leq \frac{\epsilon_{n,p}}{2}.
\end{split}
\end{equation}
Thus for large enough $n$, $\Pr(E_2)\geq 1-\epsilon_{n,p}/2$, and thus
\[
\Pr\{E_1\cap E_2\}\geq \Pr(E_1\cap E_2|E_2)\Pr(E_2)\geq 1-\epsilon_{n,p}.
\]
By the results in Step I and Step II, we can finally get
\[
\inf_{\bOmega\in\cU(q,s_0(p),M)}\Pr\left\{\|\hat\bOmega-\bOmega_0\|_2\leq Cs_0(p)\left(\frac{\log p}{n}\right)^{(1-q)/2}\right\}\geq 1-\epsilon_{n,p},
\]
which concludes the theorem.

\section*{Proof of Theorem \ref{theorem:3}}
By the property (ii) of the thresholding function $\tau_\lambda(\cdot)$, and the result in (\ref{equ:S1}), we have that as $n\rightarrow\infty$
\[
\begin{split}
\Pr\{\hat\omega_{ij}\neq 0,\omega_{ij}^0=0 \ \text{for some} \ i,j\}\leq& \Pr\left\{\bmax_{i,j}|\hat\gamma^M_{ij}-\omega_{ij}^0|\geq \lambda_{ij}\right\} \\
=&\Pr\left\{\bmax_{i,j}\frac{|\hat{\gamma}_{ij}^M-\omega_{ij}^0|}{\sqrt{\hat\gamma_{ii}^M\hat\gamma_{jj}^M}}\geq\lambda\sqrt{\frac{\log p}{n}}\right\}\\
\leq &\frac{\epsilon_{n,p}}{2}\rightarrow0,
\end{split}
\]
which concludes the result in (\ref{equ:theorem31}).

By (\ref{equ:S222}) and Theorem \ref{theorem:2}, we have with probability at least $1-\epsilon_{n,p}$,
\[
|\hat\gamma^M_{ii}\hat\gamma^M_{jj}-\omega_{ii}^0\omega_{jj}^0|\leq {\frac{3}{4}}{\zeta^2}.
\]
Thus, with probability at least $1-\epsilon_{n,p}$,
\[
\left|\sqrt{\hat\gamma^M_{ii}\hat\gamma^M_{jj}}-\sqrt{\omega_{ii}^0\omega_{jj}^0}\right|
=\frac{|\hat\gamma^M_{ii}\hat\gamma^M_{jj}-\omega_{ii}^0\omega_{jj}^0|}{\sqrt{\hat\gamma^M_{ii}\hat\gamma^M_{jj}}
+\sqrt{\omega_{ii}^0\omega_{jj}^0}}\leq \frac{\frac{3}{4}\zeta^2}{\zeta+\frac{\zeta}{2}}=\frac{1}{2}\zeta.
\]
By the property (iii) of the thresholding function $\tau_{\lambda}(\cdot)$,
\[
\Pr\left\{\text{sgn}(\hat\omega_{ij})\neq\text{sgn}(\omega_{ij}^0), \omega_{ij}^0\neq 0 \ \text{for some}\ (i,j)\right\}\leq \Pr\left\{|\hat\gamma^M_{ij}-\omega_{ij}^0|\geq |\omega_{ij}^0|-\lambda_{ij} \ \text{for some} \ i,j\right\},
\]
and
\[
\begin{split}
|\omega_{ij}^0|-\lambda_{ij}\geq& C\sqrt{\frac{\log p}{n}}\sqrt{\omega_{ii}^0\omega^0_{jj}}-\lambda\sqrt{\frac{\log p}{n}}\left(\sqrt{\hat\gamma^M_{ii}\hat\gamma^M_{jj}}-\sqrt{\omega_{ii}^0\omega_{jj}^0}+\sqrt{\omega_{ii}^0\omega_{jj}^0}\right)\\
\geq&(C-\lambda)\sqrt{\frac{\log p}{n}}\zeta-\lambda\sqrt{\frac{\log p}{n}}\frac{\zeta}{2}=(C-\frac{3}{2}\lambda)\sqrt{\frac{\log p}{n}}\zeta
\end{split}
\]
for all $(i,j)\in\{1,\ldots,p\}$. Further by Theorem \ref{theorem:1}, we yield
\[
\Pr\left\{\text{sgn}(\hat\omega_{ij})\neq\text{sgn}(\omega_{ij}^0), \omega_{ij}^0\neq 0 \ \text{for some}\ (i,j)\right\}\leq \epsilon_{n,p},
\]
which, together with (\ref{equ:theorem31}), concludes the result.

\clearpage

\section{Simulation Results for Case 3 and Case 4}
\begin{table}[ht]
\caption{Simulation results for Case 3, the values in the parenthesis are the standard errors.}  
\bigskip
\centering 
\begin{tabular}{c c c c c c c c}  
\toprule[2pt]
   $p$  &  $\widehat{\Omega}_{rcec}$  & $\widehat{\Omega}_{oracle}$ & $\widehat{\Omega}_{coat}$ & & $\widehat{\Omega}_{cclasso}$           &$\widehat{\Omega}_{sparcc}$ &   $\widehat{\Omega}_{rebacca}$       \\ [0.5ex]
\hline  
          &   \multicolumn{7}{c}{\textbf{Matrix $L_1$ norm loss}}\\[0.5ex]
      50  &8.479(0.575)  &8.304(0.598)  &10.106(3.760) &  &13.412(5.861) &68.682(48.043) &21.207(13.489)\\
     100  &7.988(0.375)  &8.267(0.675)  &12.482(10.575) &  &22.796(49.184) &185.347(238.630) &44.659(66.531)\\
     200  &8.140(0.350)  &8.278(0.377)  &12.451(5.170) &  &19.007(9.491) &338.911(233.261) &75.528(47.222)\\
          &   \multicolumn{7}{c}{\textbf{Spectral norm loss}}\\[0.5ex]
      50  &6.562(0.591)  &6.367(0.793)  &8.350(4.085) &  &13.681(4.659) &34.699(24.697) &18.143(7.575)\\
     100  &6.287(0.608)  &6.406(0.836)  &10.982(10.889) &  &17.236(18.035) &76.923(100.631) &27.968(29.828)\\
     200  &6.382(0.579)  &6.493(0.473)  &11.112(5.654) &  &16.274(6.085) &123.675(93.907) &35.512(16.766)\\
          &    \multicolumn{7}{c}{\textbf{Matrix Fronbenius norm loss}}\\[0.5ex]
      50  &13.984(2.466) &14.340(2.661) &22.054(8.668)&  &24.560(9.708) &45.674(24.209) &30.143(12.106)\\
     100  &18.789(3.103) &19.075(2.954) &34.592(21.669)&  &37.206(25.768) &103.251(96.762) &54.998(41.144)\\
     200  &25.621(3.244) &25.741(3.239) &46.755(16.341)&  &49.241(17.392) &182.748(89.221) &90.703(30.824)\\
          &    \multicolumn{7}{c}{\textbf{True positive rate}}\\[0.5ex]
      50  &0.500(0.055)  &0.486(0.052)  &0.645(0.078) &  &0.641(0.238) &1.000(0.000) &0.312(0.080)\\
     100  &0.508(0.040)  &0.493(0.041)  &0.620(0.095) &  &0.529(0.195) &1.000(0.000) &0.436(0.112)\\
     200  &0.507(0.032)  &0.497(0.033)  &0.633(0.077)&  &0.479(0.212) &1.000(0.000) &0.492(0.093)\\
          &   \multicolumn{7}{c}{\textbf{False positive rate}}\\[0.5ex]
      50  &0.041(0.022)  &0.022(0.017)  &0.135(0.051) &  &0.223(0.097) &1.000(0.000) &0.029(0.009)\\
     100  &0.010(0.006)  &0.007(0.005)  &0.065(0.028) &  &0.061(0.029) &1.000(0.000) &0.022(0.005)\\
     200  &0.003(0.001)  &0.003(0.001)  &0.027(0.009)&  &0.019(0.010) &1.000(0.000) &0.024(0.004)\\
\bottomrule[2pt]
\end{tabular}
\label{table:Case3}  
\end{table}

\begin{table}[ht]
\caption{Simulation results for Case 4, the values in the parenthesis are the standard errors.}  
\bigskip
\centering 
\begin{tabular}{c c c c c c c c}  
\toprule[2pt]
   $p$  &  $\widehat{\Omega}_{rcec}$  & $\widehat{\Omega}_{oracle}$ & $\widehat{\Omega}_{coat}$ & & $\widehat{\Omega}_{cclasso}$           &$\widehat{\Omega}_{sparcc}$ &   $\widehat{\Omega}_{rebacca}$       \\ [0.5ex]
\hline  
          &   \multicolumn{7}{c}{\textbf{Matrix $L_1$ norm loss}}\\[0.5ex]
      50  &8.386(0.480)  &8.205(0.621)  &10.589(5.680) &  &13.845(9.293) &79.065(73.171) &21.707(12.012)\\
     100  &8.183(0.363)  &8.312(0.540)  &11.349(5.454) &  &16.523(15.685) &162.531(134.155) &38.722(24.523)\\
    200  &8.344(0.339)  &8.437(0.332)  &12.751(6.346) &  &20.318(17.890) &352.603(255.300) &74.738(59.750)\\
          &   \multicolumn{7}{c}{\textbf{Spectral norm loss}}\\[0.5ex]
      50  &6.510(0.474)  &6.290(0.685)  &8.995(6.100) &  &14.450(7.173) &39.246(32.827) &19.032(8.702)\\
     100  &6.541(0.648)  &6.505(0.689)  &9.980(5.863) &  &15.508(8.370) &69.609(63.409) &25.306(24.523)\\
    200  &6.680(0.462)  &6.643(0.453)  &11.432(6.851) &  &16.749(7.564) &127.982(98.690) &36.044(19.331)\\
          &    \multicolumn{7}{c}{\textbf{Matrix Fronbenius norm loss}}\\[0.5ex]
      50  &13.285(2.251)  &13.526(2.450)  &22.257(10.201) &  &25.055(11.637) &49.287(31.515) &30.433(12.963)\\
     100  &18.145(3.134) &18.575(3.226) &31.771(13.214)&  &34.179(15.189) &93.807(59.088) &50.869(21.895)\\
    200  &24.000(3.700)  &24.305(3.567)  &44.124(16.654) &  &46.803(17.890) &180.864(90.821) &87.930(30.060)\\
          &    \multicolumn{7}{c}{\textbf{True positive rate}}\\[0.5ex]
      50  &0.483(0.060)  &0.478(0.062)  &0.594(0.114) &  &0.522(0.275) &1.000(0.000) &0.287(0.092)\\
     100  &0.489(0.046)  &0.485(0.044)  &0.624(0.068) &  &0.493(0.211) &1.000(0.000) &0.426(0.093)\\
    200  &0.497(0.033)  &0.487(0.036)  &0.619(0.073) &  &0.432(0.237) &1.000(0.000) &0.482(0.089)\\
          &   \multicolumn{7}{c}{\textbf{False positive rate}}\\[0.5ex]
      50  &0.038(0.023)  &0.017(0.014)  &0.114(0.053) &  &0.183(0.097) &1.000(0.000) &0.026(0.008)\\
     100  &0.010(0.005)  &0.007(0.004)  &0.062(0.025) &  &0.059(0.033) &1.000(0.000) &0.024(0.005)\\
    200  &0.003(0.001)  &0.003(0.001)  &0.026(0.008) &  &0.178(0.010) &1.000(0.000) &0.023(0.004)\\
\bottomrule[2pt]
\end{tabular}
\label{table:Case4}  
\end{table}

\end{appendices}

\clearpage

\section*{Acknowledgements}
Yong He's research is partially supported by the grant of the National Science Foundation of China (NSFC 11801316), Natural Science Foundation of Shandong Province (ZR2019QA002). Xinsheng Zhang's work is partially supported by the grant of the National Science Foundation of China (NSFC 11971116). The authors would like to thank Dr. Yuanpei Cao for providing the real dataset analyzed in this article.
\bibliographystyle{model2-names}
\bibliography{Ref}

\end{document}